\documentclass[letterpaper, 10pt, conference]{latex_classes/ieeeconf}      
\IEEEoverridecommandlockouts                              
\usepackage{algorithmic}
\usepackage{graphicx}
\usepackage{algorithm}
\usepackage{hyperref}
\usepackage{textcomp}
\def\BibTeX{{\rm B\kern-.05em{\sc i\kern-.025em b}\kern-.08em
    T\kern-.1667em\lower.7ex\hbox{E}\kern-.125emX}}
\usepackage{amsmath, amssymb, amsfonts}
\usepackage[capitalise]{cleveref}

\usepackage{xcolor}

\crefname{equation}{}{}
\crefname{assumption}{Assumption}{Assumptions}

\newtheorem{assumption}{Assumption}
\newtheorem{definition}{Definition}
\newtheorem{lemma}{Lemma}
\newtheorem{theorem}{Theorem}
\newtheorem{corollary}{Corollary}

\AtBeginDocument{%
    \setlength{\abovedisplayskip}{1.3ex plus 2pt minus 1pt}%
    \setlength{\belowdisplayskip}{1.3ex plus 2pt minus 1pt}%
    \setlength{\belowdisplayshortskip}{1.3ex plus 2pt minus 1pt}%
    \setlength{\textfloatsep}{1.25\baselineskip plus 0.15\baselineskip minus 0.2\baselineskip}%
    \setlength{\abovecaptionskip}{0.35\baselineskip}%
}

\newcommand{\bbX}{\mathbb{X}}
\newcommand{\bbU}{\mathbb{U}}
\newcommand{\bbR}{\mathbb{R}}
\newcommand{\bbI}[2]{\mathbb{I}_{#1}^{#2}}
\newcommand{\bfu}{\mathbf{u}}

\newcommand{\calU}{\mathcal{U}}
\newcommand{\calZ}{\mathcal{Z}}

\definecolor{myDarkBlue}{RGB}{0,0,139}

\title{\LARGE \bf
    Economic Model Predictive Control with Policy-Guided \\Terminal Ingredients
}

\author{Salim Msaad and Robert D. McAllister
\thanks{The authors are with the Delft Center for Systems and Control, 
    Delft University of Technology, 2628 Delft, The Netherlands (e-mail: 
    \{s.msaad, r.d.mcallister\}@tudelft.nl).}%
}

\begin{document}

\bstctlcite{BSTcontrol}

\maketitle
\thispagestyle{empty}
\pagestyle{empty}

\begin{abstract}

Conventional designs for model predictive control typically rely on terminal costs and constraints derived from a steady state to guarantee closed-loop stability and performance. 
However, this dependence on a steady-state assumption limits the applicability of this control method to systems in which such a fixed operating point is either not available or not desirable. 
This work introduces a novel framework, termed policy-guided MPC, to address this limitation. 
Our approach constructs terminal costs and constraints using a known sub-optimal control policy. 
Specifically, the terminal region is defined around a center determined by a rollout of the policy, and a penalty on deviation from this center is used to define the terminal cost.
This method obviates the need for a steady state or reference trajectory.
Closed-loop performance guarantees are established relative to the guiding policy, for both finite and infinite horizon problems.
The effectiveness of the proposed framework is demonstrated through numerical simulations on an energy management example.

\end{abstract}

\section{Introduction}
Model predictive control (MPC) is widely applied in industry due to its capability to manage constraints and multivariable systems while optimizing performance \cite{rawlingsModelPredictiveControl2017}. 
At each time step, MPC solves an optimal control problem over a finite horizon and applies only the first control input to the system.
Depending on the cost function used in the optimal control problem, MPC can be classified into two main formulations.
The first formulation, known as tracking MPC, is used when a specific reference setpoint or trajectory represents the desired operating condition and the goal is to minimize deviations from this reference.
The second formulation, known as economic MPC, is typically considered when the goal is to optimize a specific performance metric, such as energy consumption, production throughput, or economic profit.
In economic MPC, the cost function directly reflects this performance metric, rather than a tracking error.
This allows the controller to explicitly optimize the operational objectives of the system, making it particularly well suited for applications in which the optimal operation does not consist of a fixed steady state or trajectory.

While tracking MPC typically assumes a positive definite stage cost, economic MPC does not necessarily satisfy this assumption. 
This fundamental difference gives rise to distinct theoretical properties and challenges for economic MPC.
The first closed-loop stability results for economic MPC were based on the optimal steady state of the system.
In \cite{rawlingsUnreachableSetpointsModel2008}, the optimal steady state is determined and included as a terminal equality constraint in the optimal control problem. 
Stability is ensured under the assumptions of model linearity and convex stage cost.
Diehl \textit{et al.} \cite{diehlLyapunovFunctionEconomic2011} employ the same methodology and identify a Lyapunov function for problems satisfying strong duality of the steady-state optimization problem.
This analysis is further generalized by Angeli \textit{et al.} \cite{angeliAveragePerformanceStability2012} to economic MPC problems meeting dissipativity conditions on the system dynamics. The authors also provide results for periodic reference trajectories which are later extended to generic time-varying reference trajectories in \cite{risbeckEconomicModelPredictive2020}.
In \cite{amritEconomicOptimizationUsing2011}, the authors relax the equality constraint on the optimal steady state, replacing it with a terminal cost and a terminal region constraint.

An alternative method to terminal constraint design consists in the use of generalized
terminal ingredients, where the terminal constraint is left as a free variable, rather than being fixed a priori.
Building on the work of Ferramosca \textit{et al.} \cite{ferramoscaMPCTrackingOptimal2008}, Fagiano and Teel \cite{fagianoGeneralizedTerminalState2013} propose a method that allows any admissible fixed point to serve as a terminal constraint.
The framework was later extended in \cite{mullerEconomicModelPredictive2013} and \cite{mullerPerformanceEconomicModel2014} through the use of a self-tuning terminal cost. 
The use of generalized terminal ingredients avoids computing the optimal steady state and enlarges the feasibility set.
Nevertheless, this approach still relies on steady states for terminal constraint construction, and performance guarantees remain relative to the optimal steady state.

In the absence of terminal constraints, Grüne \cite{gruneEconomicRecedingHorizon2013} establishes closed-loop stability and near-optimal performance guarantees under stronger assumptions based on controllability and turnpike properties. Recent results extend this approach to optimal infinite horizon trajectories provided these trajectories satisfy an infinite horizon turnpike property \cite{grune2020economic}. In both cases, sufficiently long horizons are required for these performance guarantees.  

All of the results mentioned above form a strong foundation for stability and performance guarantees in economic MPC. 
However, the design of terminal costs and constraints around a steady state can suffer from performance limitations.
Such limitations occur when even the optimal steady state yields unsatisfactory performance, performing significantly worse than other feasible operating trajectories.
This work addresses these issues by utilizing an existing control policy.
For a large number of applications, there exists---or it is easy to design---a sub-optimal control policy that achieves modest performance.
Of particular interest are applications where this policy is significantly better than maintaining the optimal steady state.
This policy may be obtained from existing control algorithms, expert heuristics, or reinforcement learning.
We propose a novel approach, termed policy-guided MPC (PG-MPC), to construct terminal costs and constraints based on this policy.
In contrast to existing methods to design terminal costs/constraints, this approach provides closed-loop performance guarantees relative to the chosen \textit{policy} rather than a steady state. 

The main contributions of this work are threefold. 
First, we introduce the policy-guided MPC framework.  
Second, we establish closed-loop performance guarantees relative to the guiding policy. 
Third, we demonstrate the effectiveness of the proposed approach through numerical simulations on an energy management example with time-varying prices.

\section{Policy-Guided Model Predictive Control}

\subsection{Problem Formulation}
We consider the discrete-time dynamical system
\begin{equation}\label{eq:f}
    x^+ = f(x,u) \qquad f:\bbX\times\bbU\rightarrow\bbR^n
\end{equation}
in which $x\in\bbX\subseteq\bbR^n$ is the state and $u\in\bbU\subseteq\bbR^m$ is the input. 
Here $\bbX$ and $\bbU$ denote the state and input constraint sets, respectively.
For this system, we consider the stage cost $\ell:\bbX\times\bbU\rightarrow\bbR_{\geq 0}$ to quantify the relative performance of different state/input combinations. 
We make the following standard regularity assumption for this system and stage cost.

\begin{assumption}[Continuous functions and closed sets]\label{as:cont}
    The functions $f:\bbX\times\bbU\rightarrow\bbR^n$ and $\ell:\bbX\times\bbU\rightarrow\bbR_{\geq 0}$ are continuous, $\bbX$ is closed, and $\bbU$ is compact.
\end{assumption}

The standard optimal control problem is to select an input trajectory that minimizes the infinite-horizon discounted cost function from an initial state $x(0)$:
\begin{equation*}
    \sum_{k=0}^{\infty} \gamma^k \ell(x(k), u(k)) \qquad \textnormal{s.t.} \qquad \cref{eq:f}
\end{equation*}
in which $\gamma\in(0,1]$ is the discount factor.

\subsection{Available policy}
We assume that an initial sub-optimal policy $\pi(\cdot)$ is already available for the system. 
This policy may be obtained from existing control algorithms, prior domain expertise, or reinforcement learning. 
However, a crucial practical requirement in the subsequently proposed MPC framework is that evaluating this policy is computationally efficient. 
Thus, constructing this initial policy using methods that are computationally demanding to evaluate, such as MPC, may not be desirable.

Let $\phi_{\pi}(k;x)$ denote the state at time $k$ given the initial state $x$ for the autonomous, closed-loop system $x^+=f(x,\pi(x))$.
To streamline notation, we define the stage cost along this closed-loop trajectory as
\begin{equation*}
    \ell_\pi(k,x) := \ell(\phi_\pi(k;x),\pi(\phi_\pi(k;x))).
\end{equation*}
For this policy, there exists a corresponding value function $J_{\pi}(x)$ defined as
\begin{equation*}
    J_{\pi}(x) := \sum_{k=0}^{\infty}\gamma^k\ell(\phi_{\pi}(k;x),\pi(\phi_{\pi}(k;x)) )
\end{equation*}
and, by definition, satisfying the Bellman equation
\begin{equation*}
    J_{\pi}(x) = \ell(x,\pi(x)) + \gamma J_{\pi}(f(x,\pi(x))) \quad \forall x\in\bbX
\end{equation*}
Moreover, we make the following assumption.
\newline
\begin{assumption}[Policy assumptions]\label{as:finite_and_lipschitz}
    \,
    \begin{itemize}
        \item The policy is state and input constraint-satisfying, i.e.
        \begin{equation*}
            \pi(x)\in\bbU, \quad f(x,\pi(x))\in\bbX \quad \, \forall x\in\bbX.
        \end{equation*}
        \item The function $J_{\pi}(x)$ is finite, i.e.
        \begin{equation*}
            J_{\pi}(x) < \infty \quad \forall x\in\bbX
        \end{equation*}
        and Lipschitz continuous, i.e., there exists a constant $L>0$ such that:
        \begin{equation*}
            \lvert J_{\pi}(x_1) - J_{\pi}(x_2)\rvert \leq L \lvert x_1 - x_2\rvert \quad\, \forall x_1,x_2\in\bbX
        \end{equation*}
    \end{itemize}
\end{assumption}

Here and in all subsequent sections, $\lvert\cdot\rvert$ represents the Euclidean norm.

\subsection{Policy-Guided MPC Framework}\label{subsec:pgmpc_framework}
In conventional MPC methods, terminal costs and constraints are constructed based on a steady state or fixed reference trajectory, defined a priori for the system. 
Policy-guided MPC instead actively constructs and modifies these terminal costs and constraints based on the available policy.
Specifically, we define the terminal cost $V_f(x,x_f)$ and the terminal constraint set $\bbX_f(x_f)$ with respect to a central point $x_f\in\bbX$ that is updated at each time step using the available policy $\pi(\cdot)$.
Using two design parameters $\tau\geq0$ and $\rho>0$, we define:
\begin{equation*}
    V_f(x,x_f) := \rho \lvert x - x_f\rvert
\end{equation*}
\begin{equation*}
    \bbX_f(x_f) := \{x\in\bbX \mid \lvert x - x_f\rvert \leq \tau\}
\end{equation*}
Moreover, we introduce the policy approximation error $d$ as:
\begin{equation*}
    d := \max\{0,(L-\rho)\tau\}
\end{equation*}
The role of $d$ in the PG-MPC framework is established in \cref{subsec:cost_decrease_ineq}, where we show that smaller values of $d$ yield tighter performance bounds relative to the guiding policy.
Notably, achieving $d = 0$ does not require knowledge of $J_\pi(x)$ or the Lipschitz constant $L$; it suffices to choose a sufficiently large $\rho$ or to set $\tau = 0$.
However, excessively large $\rho$ or $\tau = 0$ reduces PG-MPC's freedom to optimize.


For a horizon $N\geq 1$, let $\hat{\phi}(k;x,\mathbf{u})$ denote the predicted state at time $k\in\bbI{0}{N}$, where $\bbI{a}{b}:=\{a,a+1,\dots,b\}$ denotes the integer interval from $a$ to $b$, for the dynamics in \cref{eq:f} given the initial state $x\in\bbX$ and the input trajectory $\bfu=(u(0),u(1),\dots,u(N-1))\in\bbU^N$. 
We then define the PG-MPC objective function $V_N:\bbX\times\bbX\times\bbU^N\rightarrow\bbR$ as:
\begin{equation*}
\begin{aligned}
    V_N(x,x_f,\bfu) := {} \sum_{k=0}^{N-1} \gamma^k\ell(&\hat{\phi}(k;x,\bfu),u(k)) \\
    & + \gamma^N V_f(\hat{\phi}(N;x,\bfu),x_f)
\end{aligned}
\end{equation*}
    
The input constraint set is defined as all feasible input trajectories that drive the system from the initial state $x$ to the terminal region $\bbX_f(x_f)$, while satisfying the state constraints at all intermediate time steps:
\begin{equation*}
    \calU_N(x,x_f) := \left\{\bfu\in\bbU^N \;\middle|\; \begin{array}{l}
        \hat{\phi}(k;x,\bfu)\in\bbX \quad \forall k\in\bbI{1}{N-1}, \\
        \hat{\phi}(N;x,\bfu)\in\bbX_f(x_f)
    \end{array}\right\}
\end{equation*}
The feasible set for the pair $(x,x_f)$ is then defined as all pairs that result in a non-empty input constraint set:
\begin{equation*}
    \calZ_N := \{(x,x_f)\in\bbX\times\bbX \mid \calU_N(x,x_f)\neq\emptyset\}.
\end{equation*}
Policy-guided MPC then solves the following optimization problem at each time step:
\begin{equation}\label{eq:pgmpc}
    V_N^0(x,x_f) := \min_{\bfu\in\calU_N(x,x_f)} V_N(x,x_f,\bfu)
\end{equation} 
and the optimal solution to this problem is denoted as
\begin{equation*}
    \bfu^0(x,x_f) \!:= \!(u^0(0;x,x_f),u^0(1;x,x_f),\!\dots\!,u^0(N\!-\!1;x,x_f))
\end{equation*}
Given this optimal solution, only the first input is applied to the system at each time step. We define the PG-MPC control policy as:
\begin{equation*}
    \kappa(x,x_f) = u^0(0;x,x_f)
\end{equation*}
At each time step, we update $x_f$ by rolling out the policy $\pi(\cdot)$ from the last state of the optimal predicted trajectory.
Thus, we consider the following closed-loop dynamics for the pair $(x,x_f)$:
\begin{equation}\label{eq:(x,x_f)_cl}
    \begin{aligned}
        x^+ = h_1(x,x_f) := f(&x,u^0(0;x,x_f)) \\
        x_f^+ = h_2(x,x_f) := f(&x^0(N),\pi(x^0(N))) \\
        &x^0(N) = \hat{\phi}(N;x,\bfu^0(x,x_f))
    \end{aligned}
\end{equation}

The closed-loop dynamics of the pair $(x,x_f)$ are illustrated in \cref{fig:cl_dia}.
Note that the next central point $x_f^+$ is obtained by rolling out the policy $\pi(\cdot)$ from the last predicted state of the \textit{optimal trajectory}, rather than from the current central point $x_f$.
For the first time step, however, the initial central point $x_f(0)$ is defined as the state obtained by rolling out the policy $\pi(\cdot)$ from the initial state $x(0)$ for $N$ steps, i.e., we have that
\begin{equation*}
    x_f(0) = \phi_{\pi}(N;x(0)).
\end{equation*}

\begin{figure}
    \centering
    \includegraphics[width=0.5\textwidth]{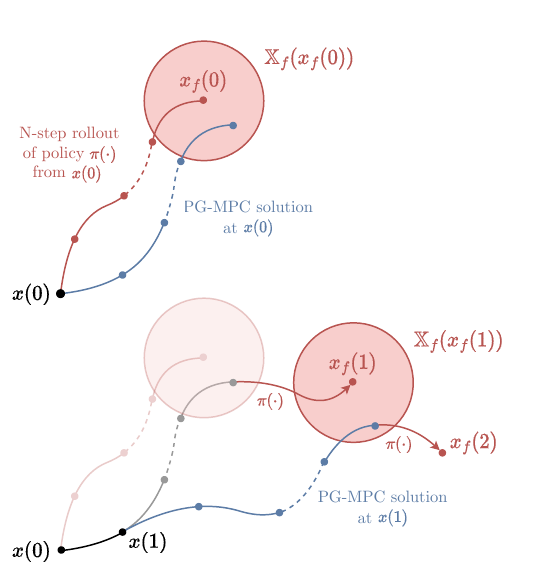}
    \caption{Illustration of the closed-loop dynamics of the pair $(x,x_f)$ for policy-guided MPC. For the first time step, the initial central point $x_f$ is defined as the state obtained by rolling out the policy $\pi(\cdot)$ from the initial state $x$ for $N$ steps. The next state $x^+$ is obtained by applying the first input of the optimal input trajectory, while the next central point $x_f^+$ is obtained by rolling out for one step the policy $\pi(\cdot)$ from the last predicted state of the optimal trajectory. This is repeated for all subsequent time steps.}
    \label{fig:cl_dia}
\end{figure}

\section{Closed-loop guarantees}\label{sec:closed_loop_guarantees}
To streamline notation, we define the extended state as $z:=(x,x_f)\in\calZ_N$ and describe the dynamics of $z$ via the difference equation
\begin{equation}\label{z_cl}
    z^+ = (x^+,x_f^+) = h(z) := (h_1(x,x_f),h_2(x,x_f))
\end{equation}

We now consider the closed-loop extended state from the initial state $x\in\bbX$.
Given the initial state $x\in\bbX$ and the dynamics in \cref{z_cl}, let $\psi(k;x)$ denote the closed-loop extended state at time $k$, i.e., $\psi(\cdot)$ satisfies
\begin{align*}
    &\psi(0;x) =(x,\phi_{\pi}(N;x)) \\
    &\psi(k+1;x) = h(\psi(k;x)) \qquad \forall k\in\bbI{0}{\infty}
\end{align*}
Thus, we have the closed-loop state at time $k$, denoted as $\phi_{\kappa}(k;x)$ and defined such that
\begin{equation*}
    (\phi_{\kappa}(k;x),\,\cdot\,) = \psi(k;x).
\end{equation*}
To streamline notation, we define the stage cost along the PG-MPC closed-loop trajectory as
\begin{equation*}
    \ell_\kappa(k,x) := \ell(\phi_\kappa(k;x),\kappa(\psi(k;x))).
\end{equation*}

To analyze the closed-loop properties of policy-guided MPC, we frequently make use of the following candidate solution for the optimization problem in \cref{eq:pgmpc}: 
For any $z=(x,x_f)$, we denote the candidate solution for $z^+$ as
\begin{equation}\label{eq:warm_start}
    \begin{aligned}
        \tilde{\bfu}^+ = \zeta(z) := \Big(u^0(1;z),\dots,u^0&(N-1;z), \\
        &\pi\big(\hat{\phi}(N;x,\bfu^0(z))\big)\Big)
    \end{aligned}
\end{equation}

To guarantee that the MPC optimization problem remains feasible at all times, we next establish that $\calZ_N$ is positively invariant.

\begin{definition}[Positive invariance]
    The set $\calZ$ is positive invariant for the system $z^+=f(z)$ if $f(z)\in\calZ$ for all $z\in\calZ$.
\end{definition}

\begin{lemma}\label{lem:posinv_and_welldef}
    If \cref{as:cont,as:finite_and_lipschitz} hold, then for each $z\in\calZ_N$ a solution to the optimization problem in \cref{eq:pgmpc} exists, the set $\calZ_N$ is positive invariant for the system in \cref{z_cl}, and the closed-loop (extended) state trajectories $\psi(k;x)$ are well defined for all $x\in\bbX$ and $k\in\bbI{0}{\infty}$.
\end{lemma}

\begin{proof}
We have that $V_f(\cdot)$ is continuous and $\bbX_f$ is compact by definition. Combined with regularity conditions in \cref{as:cont}, we satisfy the standard requirements to ensure that a solution to the MPC optimization problem in \cref{eq:pgmpc} exists for all $z\in\calZ_N$ \cite[Prop. 2.4]{rawlingsModelPredictiveControl2017}. 

Choose any $z\in\calZ_N$. 
For $z^+=h(z)$, we consider the candidate solution $\tilde{\bfu}^+ = \zeta(z)$. 
Clearly, $\tilde{\bfu}^+\in\bbU^N$.
Since the inputs $u^0(1;z),\dots,u^0(N-1;z)$ were feasible for $z$, the intermediate predicted states satisfy $\hat{\phi}(k;x^+,\tilde{\bfu}^+)=\hat{\phi}(k+1;x,\bfu^0(z))\in\bbX$ for all $k\in\bbI{1}{N-1}$.
Since the first $N-2$ inputs are the same as the solution for $z$, we have that
\begin{equation*}
    x^0(N) := \hat{\phi}(N;x,\bfu^0(z)) = \hat{\phi}(N-1;x^+,\tilde{\bfu}^+)
\end{equation*}
and therefore
\begin{equation*}
    \hat{\phi}(N;x^+,\tilde{\bfu}^+) = f(x^0(N),\pi(x^0(N))) = x_f^+\in\bbX_f(x_f^+).
\end{equation*}
Furthermore, $x_f^+\in\bbX$ since the policy satisfies the state constraints by \cref{as:finite_and_lipschitz}.
Thus, we have that $\tilde{\bfu}^+\in\calU_N(z^+)$ and therefore $z^+\in\calZ_N$. 
Since the choice of $z\in\calZ_N$ was arbitrary, we have that $\calZ_N$ is positive invariant. 

Choose any initial $x\in\bbX$. 
Then, we have that the initial $x_f$ is defined as $x_f=\phi_{\pi}(N;x)$. 
Thus, we have that this terminal state can be reached by using the same policy rollout. 
Specifically, we can define the input trajectory
\begin{equation*}
    \tilde{\bfu} = 
    (\pi(\phi_{\pi}(0;x)),\pi(\phi_{\pi}(1;x)),\dots, 
    \pi(\phi_{\pi}(N-1;x)))
\end{equation*}
By recursive application of \cref{as:finite_and_lipschitz}, we have that $\tilde{\bfu}\in\bbU^N$ and $\hat{\phi}(k;x,\tilde{\bfu})=\phi_{\pi}(k;x)\in\bbX$ for all $k\in\bbI{1}{N}$.
In particular, $\hat{\phi}(N;x,\tilde{\bfu})=\phi_{\pi}(N;x)=x_f\in\bbX_f(x_f)$.
Thus, $\tilde{\bfu}\in\calU_N(x,x_f)$ and therefore $z=(x,x_f)\in\calZ_N$. 
Therefore, $\psi(0;x)\in\calZ_N$ and $\kappa(\psi(0;x))$ is well defined.
Since $\calZ_N$ is positive invariant, PG-MPC remains feasible and $\psi(k;x)$ is well defined.
\end{proof}

\subsection{Cost decrease inequality}\label{subsec:cost_decrease_ineq}
To simplify the subsequent analysis, we define the shifted cost as the sum of the optimal MPC cost and the value function for the policy $\pi(\cdot)$ starting at the central point $x_f$: 
\begin{equation}\label{eq:shifted_cost}
    \bar{V}_N^0(x,x_f) := V_N^0(x,x_f) + \gamma^N J_{\pi}(x_f)
\end{equation}

For this shifted cost, we can establish the following cost decrease inequality.
\begin{lemma}[Cost decrease inequality]
    If \cref{as:cont,as:finite_and_lipschitz} hold, then we have that
    \begin{equation}\label{eq:cost_decrease}
        \gamma\bar{V}_N^0(x^+,x_f^+) \leq  \bar{V}_N^0(x,x_f) - \ell(x,\kappa(x,x_f)) + \gamma^Nd
    \end{equation}
    for all $(x,x_f)\in\calZ_N$. 
\end{lemma}

\begin{proof} 
    For any $(x,x_f)\in\mathcal{Z}_N$, we denote the warm-start for $(x^+,x_f^+)$ as in \cref{eq:warm_start}.
    We therefore have that
    \begin{equation}\label{eq:V0_decrease}
        \begin{aligned}
            \gamma V_N(x^+,x_f^+&,\tilde{\mathbf{u}}^+) - V_N^0(x,x_f) = \\
            &-\ell(x,u^0(0;x,x_f)) \\
            &+ \gamma^N\ell(x^0(N),\pi(x^0(N)))\\ 
            &- \gamma^NV_f(x^0(N),x_f)\\
            &+\gamma^{N+1}V_f(f(x^0(N),\pi(x^0(N))),x_f^+)
        \end{aligned}
    \end{equation}
    where $x_f^+ = f(x^0(N),\pi(x^0(N)))$ and thus
    \begin{equation*}
        V_f(f(x^0(N),\pi(x^0(N))),x_f^+) = V_f(x_f^+,x_f^+) = 0
    \end{equation*}
    From \cref{as:finite_and_lipschitz}, we have that:
    \begin{equation*}
        J_{\pi}(x^0(N)) - J_{\pi}(x_f) \leq L \lvert x^0(N) - x_f\rvert
    \end{equation*}
    Moreover, since $J_{\pi}(x)$ satisfies the Bellman equation we have
    \begin{equation*}
        \gamma J_{\pi}(x_f^+) - J_{\pi}(x_f) \leq -\ell(x^0(N),\pi(x^0(N))) + L \lvert x^0(N) - x_f\rvert
    \end{equation*}
    Multiplying the above inequality by $\gamma^N$ and adding it to the cost decrease equality in \cref{eq:V0_decrease}, we obtain:
    \begin{equation*}
        \small
        \begin{aligned}
            \gamma \biggl(\!&V_N(x^+, x_f^+, \tilde{\mathbf{u}}^+) + \gamma^N J_\pi(x_f^+))\!\biggr) \!-\! \biggl( \!V_N^0(x,x_f)+\gamma^N J_\pi(x_f) \!\biggr) \\
            \leq& - \ell(x,\kappa(x,x_f)) \!-\! \gamma^N V_f(x^0(N), x_f) \!+\! \gamma^N L \lvert x^0(N) - x_f\rvert
        \end{aligned}
    \end{equation*}
    in which
    \begin{equation*}
        \begin{aligned}
            - \gamma^N V_f(x^0(N), x_f) + \gamma^N L \lvert x^0&(N) - x_f\rvert \\
            &= \gamma^N (L - \rho) \lvert x^0(N) - x_f\rvert \\
            &\leq \gamma^N d
        \end{aligned}
    \end{equation*}
    Then, by optimality and from the definition of $\bar V_N^0(\cdot)$, we obtain the cost decrease inequality in \cref{eq:cost_decrease}.
\end{proof}

\subsection{Nominal performance bounds}
We now establish nominal performance bounds for the closed-loop system under the policy-guided MPC scheme. 
These results quantify how the long-term cost of the resulting policy-guided MPC control law $\kappa(\cdot)$ compares to that of the guiding policy $\pi(\cdot)$.
In particular, we show that the cumulative cost of policy-guided MPC remains within a bounded deviation $d$ from the cost achieved by the guiding policy, with the bound scaling with the discount factor $\gamma$ and the prediction horizon $N$.

\begin{theorem}[Transient performance bound]
    If \cref{as:cont,as:finite_and_lipschitz} hold, then we have that for any $T\geq1$:
    \begin{equation}\label{eq:perf_bound}
        \begin{aligned}
            \sum_{k=0}^{T-1}\gamma^k\ell_\kappa(k,x) \leq J_{\pi}(x) + \gamma^N\sum_{k=0}^{T-1}\gamma^kd
        \end{aligned}
    \end{equation} 
\end{theorem}

\begin{proof}
    Consider any initial state $x\in\mathbb{X}$ and define $x_f:=\phi_\pi(N;x)$.
    We have that $(x,x_f)=\psi(0;x)\in\mathcal{Z}_N$ and from \cref{lem:posinv_and_welldef}, $\psi(k;x)\in\mathcal{Z}_N$ for $k\in\bbI{0}{\infty}$.
    \\
    Repeated application of \cref{eq:cost_decrease} along the closed-loop trajectory from $k=0$ to $T-1$ gives:
    \begin{equation*}
        \begin{aligned}
            &\sum_{k=0}^{T-1}\gamma^k\ell_\kappa(k,x) + \gamma^{T}\bar{V}_N^0(\psi(T;x)) \\[-0.2cm]
            &\hspace{8em} \leq \bar{V}_N^0(\psi(0;x)) + \gamma^N\sum_{k=0}^{T-1}\gamma^k d
        \end{aligned}
    \end{equation*}
    We can also use the fact that $x_f:=\phi_\pi(N;x)$ to obtain
    \begin{equation*}
        \begin{aligned}
            \bar{V}_N^0(\psi(0;x)) \leq
            \sum_{k=0}^{N-1}\gamma^k\ell_\pi(k,x) + \gamma^N J_{\pi}(\phi_\pi(N;x))
        \end{aligned}
    \end{equation*}
    Moreover, by definition of $J_{\pi}(\cdot)$ we have that
    \begin{equation*}
        \begin{aligned}
            &\sum_{k=0}^{T-1}\gamma^k\ell_\kappa(k,x) + \gamma^T\bar{V}_N^0(\psi(T;x)) \leq J_{\pi}(x) + \gamma^N\sum_{k=0}^{T-1}\gamma^kd
        \end{aligned}
    \end{equation*}
    Since $\ell(\cdot)\geq 0$ we have that $\bar{V}_N^0(\cdot)\geq 0$, and the inequality in \cref{eq:perf_bound} follows.
\end{proof}

This finite horizon performance bound can be extended to the infinite horizon case as follows.

\begin{theorem}[Asymptotic performance bounds]
    If \cref{as:cont,as:finite_and_lipschitz} hold, then we have that for $\gamma\in(0,1)$:
    \begin{equation}\label{eq:asy_perf_bound_1}
        \limsup_{T\rightarrow\infty}
        \sum_{k=0}^{T-1} \gamma^k
        \big(\ell_\kappa(k,x)-\ell_\pi(k,x)\big) \leq
        \frac{\gamma^Nd}{1-\gamma}
    \end{equation}
    and for $\gamma=1$:
    \begin{equation}\label{eq:asy_perf_bound_2}
        \limsup_{T\rightarrow\infty}
        \frac{1}{T}\sum_{k=0}^{T-1}
        \big(\ell_\kappa(k,x)-\ell_\pi(k,x)\big) \leq d
    \end{equation}
    for all $x\in\bbX$. 
\end{theorem}

\begin{proof}
    For $\gamma\in(0,1)$, we start from the performance bound in \cref{eq:perf_bound}.
    Given the definition of $J_\pi(\cdot)$, in the limit as $T\rightarrow\infty$, we obtain \cref{eq:asy_perf_bound_1}.
    For $\gamma=1$, \cref{eq:perf_bound} becomes
    \begin{equation}\label{eq:perf_bound_gamma1}
            \sum_{k=0}^{T-1}\ell_\kappa(k,x) \leq J_{\pi}(x) + Td
    \end{equation}
    We substitute the definition of $J_{\pi}(x)$ and divide by the horizon length $T$ to obtain
    \begin{equation*}
        \begin{aligned}
            \frac{1}{T}\!\sum_{k=0}^{T-1}\!\big(\ell_\kappa(k,x)-\ell_\pi(k,x)\big) \leq \frac{1}{T}J_{\pi}(\phi_\pi(T;x)) + d 
        \end{aligned}
    \end{equation*}
    Since $\ell(\cdot)\geq0$, the Bellman equation for $\gamma=1$ implies that $0\leq J_{\pi}(\phi_\pi(T;x))\leq J_{\pi}(x)<\infty$ for all $T$. Hence, $J_{\pi}(\phi_\pi(T;x))/T\rightarrow0$ as $T\rightarrow\infty$, which yields \cref{eq:asy_perf_bound_2}.
\end{proof}

Finally, we consider the special case when $\gamma=1$ and $d=0$, which occurs when the terminal cost weight $\rho$ is chosen to be at least as large as the Lipschitz constant $L$ of the value function $J_{\pi}(\cdot)$ or when the terminal region radius $\tau$ is set to zero.

\begin{corollary}[Asymptotic performance bounds for $d=0$]
    If \cref{as:cont,as:finite_and_lipschitz} hold, then we have that for 
    $\gamma=1$ and $d=0$:
    \begin{equation}\label{eq:asy_perf_bound_dzero}
        \limsup_{T\rightarrow\infty}
        \sum_{k=0}^{T-1}
        \big(\ell_\kappa(k,x)-\ell_\pi(k,x)\big) \leq 0
    \end{equation}
    for all $x\in\bbX$. 
\end{corollary}

\begin{proof}
    Setting $d=0$ in \cref{eq:perf_bound_gamma1} and substituting the definition of $J_{\pi}(x)$ yields \cref{eq:asy_perf_bound_dzero} in the limit as $T\rightarrow\infty$.
\end{proof}

\section{Numerical Example}
In this example, we analyze an optimal control problem in which the optimal steady state is not desirable.
Consider an energy system with two batteries that coordinate with the grid to satisfy demand while minimizing costs under a time-varying energy price.
The two batteries have different characteristics: a small battery with high efficiency and a large battery with lower efficiency but higher capacity.
Let $x_1(k)$ and $x_2(k)$ denote the energy stored in the small and large batteries, respectively.
The dynamics of the system evolve as follows:
\begin{equation}\label{eq:example2:dynamics}
\begin{aligned}
    x_1(k+1) =& x_1(k) + \eta_1 u_2(k) - \tfrac{1}{\eta_1} u_4(k), \\
    x_2(k+1) =& x_2(k) + \eta_2 u_3(k) - \tfrac{1}{\eta_2} u_5(k),
\end{aligned}
\end{equation}
Inputs $u_2(k)$ and $u_3(k)$ represent the energy drawn from the grid to charge the small and large batteries, while $u_4(k)$ and $u_5(k)$ represent the energy discharged from the small and large batteries to satisfy demand.
One more input, $u_1(k)$, represents the energy drawn from the grid to satisfy the demand directly.
Finally, $\eta_1$ and $\eta_2$ are the efficiencies of the small and large batteries, respectively, set to $\eta_1=0.99$ and $\eta_2=0.95$.
At each time step, the system incurs a stage cost defined as the cost of energy drawn from the grid:
\begin{equation}\label{eq:example2:stage_cost}
    \begin{aligned}
        \ell(x,u) = \big(u_1(k) + u_2(k) + u_3(k)\big)p(k)
    \end{aligned}
\end{equation}
where $p(k)$ represents the time-varying price of energy drawn from the grid.
The energy price $p(k)$ introduces time-varying dynamics; to fit the time-invariant framework in \cref{eq:f}, we augment the state with the time index $k$.
The full state is thus defined as $x(k) = (x_1(k),\, x_2(k),\, k)$.
The system is subject to demand, input, and battery-state constraints, defined as:
\begin{subequations}\label{eq:example2:constraints}
    \begin{align}
        &u_1(k)+u_4(k)+u_5(k) = E^{ref}, \label{eq:example2:constraints_a} \\
        &0 \leq x_1(k) \leq 2400, \label{eq:example2:constraints_b} \\
        &0 \leq x_2(k) \leq 33600, \label{eq:example2:constraints_c} \\
        &0 \leq u_i(k) \leq 100, \quad \forall i=1,\dots,5. \label{eq:example2:constraints_d}
    \end{align}
\end{subequations}
Here, $E^{ref}=40$ is the constant energy demand that must be met at each time step by the combination of grid energy and battery discharges.
The system is simulated over a 24-day operating horizon with a sampling time of $\Delta t = 0.5$ hours, yielding $1152$ time steps.
The objective is to minimize the total cost of energy purchased from the grid over the operating horizon while satisfying the demand and constraints at each time step.

Any steady state requires constant battery energy levels, precluding the load shifting needed to exploit the time-varying price structure.
Consequently, designing terminal costs and constraints around any steady state would be inadequate for this problem.
An alternative would be to construct a periodic trajectory a priori and use it to define a terminal equality constraint.
However, this approach requires prior knowledge of the price profile and is fragile to model mismatch and disturbances.

A standard MPC formulation that minimizes the stage cost in \cref{eq:example2:stage_cost} over a prediction horizon of $N=48$ steps (one day), subject to the dynamics in \cref{eq:example2:dynamics} and constraints in \cref{eq:example2:constraints}, is considered.
The simulated closed-loop trajectory of this MPC controller is shown in \cref{fig:energy_management}.
This controller fails to effectively utilize the batteries, resulting in high operational costs.
In particular, the MPC is myopic, as it almost never charges the large battery.
The small battery is preferred in the limited prediction horizon that it considers, as it is more efficient.

We design a rule-based guiding policy $\pi(\cdot)$ that controls only the large battery, with decisions based solely on the current energy price $p(k)$.
When the price is below a threshold $\bar{p}$, the policy charges the large battery at a fixed rate while satisfying demand from the grid. When the price exceeds $\bar{p}$, demand is met by discharging the large battery, with the grid covering any shortfall:
\begin{equation*}
    \pi(x) = \begin{cases}
        (E^{ref},\, 0,\, \bar{u}_3,\, 0,\, 0) & \text{if } p(k) < \bar{p} \\
        (E^{ref}\!-\!u_5,\, 0,\, 0,\, 0,\, u_5) & \text{if } p(k) \geq \bar{p}
    \end{cases}
\end{equation*}
where $\bar{p}=2$ is the price threshold, $\bar{u}_3=50$ is the fixed charging rate, and $u_5 = \min(E^{ref}, x_2(k))$.

PG-MPC uses the same prediction horizon and constraints as the standard MPC, with the rule-based policy determining the terminal-region center $x_f$.
The states $x$ and $x_f$ always correspond to the same time index $k$, so their time-index components coincide.
The terminal ingredients are therefore
\begin{equation*}
    \begin{gathered}
        \bbX_f(x_f) = \{x\in\bbX\mid \lvert x-x_f\rvert\leq\tau\}, \\
        V_f(x,x_f) = \rho\lvert x-x_f\rvert.
    \end{gathered}
\end{equation*}
The two parameters play complementary roles.
The radius $\tau$ delimits how far the predicted terminal state may deviate from the policy rollout: a small $\tau$ confines PG-MPC to trajectories ending close to the policy, while a large $\tau$ admits substantial deviations.
Within this region, the weight $\rho$ determines how strongly the terminal state is drawn back towards $x_f$: a large $\rho$ keeps PG-MPC closely anchored to the guiding policy, whereas a small $\rho$ allows it to exploit the economics of the prediction horizon and thereby improve upon the policy.
This suggests selecting the two sequentially: $\tau$ first, as the largest deviation from the policy rollout that is reachable within $N$ steps and acceptable for the application, and $\rho$ thereafter, trading improved performance against closer adherence to the policy.
For the simulations, we set $\tau = 1680$ and $\rho = 1.5$.
Through the policy-generated terminal center, PG-MPC overcomes the myopia of standard MPC while retaining its ability to optimize over the short prediction horizon.
The closed-loop trajectory, shown in \cref{fig:energy_management}, confirms that PG-MPC effectively utilizes both batteries and achieves substantially lower costs than either the standard MPC or the rule-based policy alone.
Importantly, this improvement is achieved without providing PG-MPC with any explicit knowledge of the price profile beyond that available to standard MPC; the additional information is embedded implicitly in the guiding policy through the price threshold $\bar p$.
As a benchmark, the globally optimal trajectory is obtained by solving the full 24-day optimal control problem.
PG-MPC incurs a total cost of $53245$, which is $5.4\%$ above the full-horizon optimum of $50527$, whereas standard MPC and the rule-based policy incur $78104$ and $70532$, i.e., $54.6\%$ and $39.6\%$ above the optimum.
Demonstrating that the combination of a simple heuristic with a short-horizon PG-MPC formulation can significantly improve performance.

\begin{figure}
    \centering
    \includegraphics[width=0.5\textwidth]{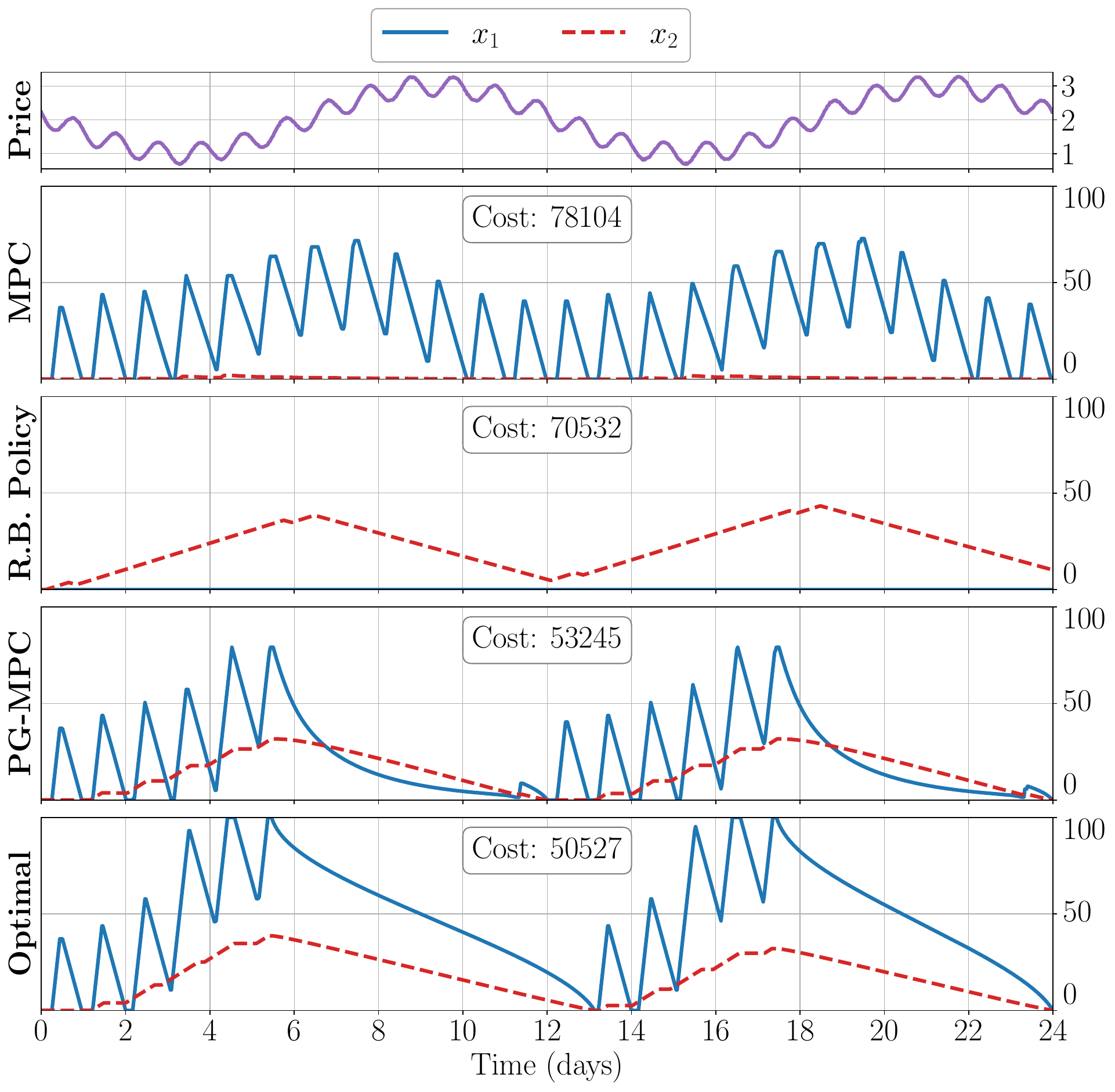}
    \caption{Simulation results for the energy management example over a 24-day horizon. The top panel shows the time-varying energy price $p(k)$. The four lower panels compare the normalized battery-state trajectories, $100x_1/2400$ and $100x_2/33600$, under standard MPC, the rule-based policy, PG-MPC, and the full-horizon optimal solution. Each controller panel reports its incurred total cost.}
    \label{fig:energy_management}
\end{figure}

\section{Conclusion}
This work introduced PG-MPC, a framework that constructs terminal costs and constraints from a given control policy rather than a steady state. We established closed-loop performance guarantees relative to the guiding policy and demonstrated the approach on an energy management example where steady-state-based designs are inadequate. Future work will investigate incorporating additional information into the terminal cost, stability in the case of a converging policy and inherent robustness.


\bibliographystyle{IEEEtran}
\bibliography{references}

\end{document}